\documentclass{llncs}

\usepackage{enumerate}
\usepackage{amssymb}
\usepackage{epsfig}
\usepackage{tikz}

\usepackage[nothing]{algorithm}
\usepackage[]{algorithmic}
\usepackage{multirow}

\newcommand{\partitle}[1]{}                        
\newcommand{\commentout}[1]{}

\newcommand{\keywords}[1]{\par\addvspace\baselineskip
\noindent\keywordname\enspace\ignorespaces#1}

\newcommand{\bi}{\begin{itemize}}
\newcommand{\ei}{\end{itemize}}

\newcommand{\be}{\begin{enumerate}}
\newcommand{\ee}{\end{enumerate}}

\newcommand{\bd}{\begin{description}}
\newcommand{\ed}{\end{description}}

\title{Suffix Tree of Alignment:\\An Efficient Index for Similar Data}

\author{
    Joong Chae Na\inst{1}
    \and
     Heejin Park\inst{2}
    \and
     Maxime Crochemore\inst{3}
    \and
     Jan Holub\inst{4}
    \and \\
     Costas S. Iliopoulos\inst{3}
    \and
     Laurent Mouchard\inst{5}
    \and
     Kunsoo Park\inst{6}
     \thanks{Corresponding author, E-mail: {\tt kpark@snu.ac.kr}}
}

\institute{
    Sejong University,
    Korea
   \and
    Hanyang University,
    Korea
   \and
    King's College London,
    UK
   \and
    Czech Technical University in Prague,
    Czech Republic
   \and
    University of Rouen,
    France
   \and
    Seoul National University,
    Korea
}

\begin{document}

\maketitle


\begin{abstract}
We consider an index data structure for similar strings.
The generalized suffix tree can be a solution for this.
The generalized suffix tree of two strings $A$ and $B$
 is a compacted trie representing all suffixes in $A$ and $B$.
It has $|A|+|B|$ leaves and can be constructed in $O(|A|+|B|)$ time.
However, if the two strings are similar,
 the generalized suffix tree is not efficient
 because it does not exploit the similarity
 which is usually represented as an alignment of $A$ and $B$.

In this paper we propose a space/time-efficient
{\em suffix tree of alignment}
which wisely exploits the similarity in an alignment.
Our suffix tree for an alignment of $A$ and $B$ has
$|A| + l_d + l_1$ leaves
where $l_d$ is the sum of the lengths of {\em all} parts of $B$ different from $A$
and $l_1$ is the sum of the lengths of {\em some} common parts of $A$ and $B$.
We did not compromise the pattern search
to reduce the space.
Our suffix tree can be searched for a pattern $P$
in $O(|P|+occ)$ time
where $occ$ is the number of occurrences of $P$ in $A$ and $B$.
We also present an efficient algorithm
to construct the suffix tree of alignment.
When the suffix tree is constructed from scratch,
the algorithm requires $O(|A| + l_d + l_1 + l_2)$ time
where $l_2$ is the sum of the lengths of other common
substrings of $A$ and $B$.
When the suffix tree of $A$ is already given,
it requires $O(l_d + l_1 + l_2)$ time.
\end{abstract}

\pagestyle{headings}  
\pagenumbering{arabic}
\setcounter{page}{1}


\section{Introduction}


\partitle{Suffix trees}

The {\em suffix tree} of a string $S$
 is a compacted trie representing all suffixes of $S$~\cite{McCreight:76,Weiner:73}.
Over the years, the suffix tree
 has not only been a fundamental data structure in the area of string algorithms
 but also it has been used for many applications in engineering and computational biology.
%
%
The suffix tree can be constructed in $O(|S|)$ time
 for a constant alphabet~\cite{McCreight:76,Ukkonen:95}
 and an integer alphabet~\cite{Farach-Colton&Ferragina:00},
 where $|S|$ denotes the length of $S$.
The suffix tree has $|S|$ leaves and requires $O(|S|)$ space.

\commentout{
\partitle{a family of ST}

Since the suffix tree was first introduced as a position tree by Weiner~\cite{Weiner:73} in 1973,
 many families and variants have been designed for various applications and circumstances.
The suffix array~\cite{KarkkainenSB06,Manber&Myers:93} is a sorted array of all suffixes
 and  the suffix automaton~\cite{Blumer&Blumer:85,Crochemore&Hancart:07},
 also known as DAWG (Directed Acyclic Word Graph),
 is an automaton representing suffixes of a given string,
 which can be seen as another representations of the suffix tree.
Moreover, compressed representations of the suffix tree
 such as the compressed suffix tree~\cite{Grossi&Vitter:05,Sadakane07TCS}
 and extensions to two-dimensional data
 such as two-dimensional suffix trees~\cite{Giancarlo:95,Kim&Na:11}
 have also been developed.
Other families and variants are
 the suffix tree of a tree~\cite{focs/Kosaraju89a},
 the generalized suffix tree~\cite{Amir&Farach:94,Gusfield:97},
 the suffix cactus~\cite{Karkkainen:95},
 the sparse suffix tree~\cite{Karkkainen&Ukkonen:96},
 the lazy suffix tree~\cite{GiegerichKS03},
 the truncated suffix tree~\cite{Na&Apostolico:03},
 the suffix binary search tree~\cite{Irving&Love:03},
 the suffix tray \& the suffix trist~\cite{ColeKL06},
 the linearized suffix tree~\cite{Kim&Kim:08}
 the geometric suffix tree~\cite{jacm/Shibuya10}, and so on.
}

\partitle{Similar data}

We consider storing
 and indexing multiple data which are very similar.
Nowadays, tons of new data are created every day.
Some data are totally original and substantially different from existing data.
Others are, however, created by modifying some existing data
 and thus they are similar to the existing data.
For example, a  new version of a source code is a modification of its previous version.
Today's backup is almost the same as yesterday's backup.
An individual Genome is more than 99\% identical to the Human reference Genome
 (the 1000 genome project~\cite{nature/1000Genomes10}).
Thus, storing and indexing similar data in an efficient way is becoming more and  more important.

Similar data are usually stored efficiently:
When new data are created, they are aligned with the existing ones.
Then, the resulting alignment shows the common and different parts of the new data.
By only storing the different parts of the new data, the similar data can be stored efficiently.

When it comes to indexing, however,
 neither the suffix tree nor any variant of the suffix tree
 uses this similarity or alignment to index similar data efficiently.
Consider the {\em generalized suffix tree}~\cite{Amir&Farach:94,Gusfield:97}
 for two similar strings $A =$ {\tt aaatcaaa} and $B =$ {\tt aaatgaaa}.
Three common suffixes {\tt aaa}, {\tt aa}, {\tt a}
 are stored twice in the generalized suffix tree.
Moreover, two similar suffixes {\tt aaatcaaa} and {\tt aaatgaaa}
 are stored in distinct leaves
 even though they are very similar.
Thus, the generalized suffix tree has $|A|+|B|$ leaves, most of which are redundant.

\partitle{Previous works}

Recently, there have been some studies concerning efficient indexes for similar strings.
M\"{a}kinen et al.~\cite{recomb/MakinenNSV09,jcb/MakinenNSV10}
 first proposed an index for similar (repetitive) strings
 using run-length encoding, a suffix array, and BWT~\cite{Burrows&Wheeler:94}.
Huang et al.~\cite{aaim/HuangLSTY10} proposed an index of size $O(n + N\log N)$ bits
 where $n$ is the total length of common parts in one string,
 $N$ is the total length of different parts in all strings.
Their basic approach is building separately data structures for common parts
 and ones for different parts between strings.
A self-index based on LZ77 compression~\cite{Ziv&Lempel:77}
 has been also developed due to Kreft and Navarro~\cite{Kreft&Navarro:XX}.
Another index based on Lemple-Ziv compression scheme is
 due to Do et al.~\cite{aaim/DoJSS12}.
They compressed sequences using a variant of the relative Lempel-Ziv (RLZ) compression
 scheme~\cite{spire/KuruppuPZ10}
 and used a number of auxiliary data structures to support fast pattern search.
Navaro~\cite{iwoca/Navarro12} gave a short survey on some of these indexes.

\partitle{Motivation}

Although these studies assume slightly different models on similar strings,
 most of them adopt classical compressed indexes to utilize the similarity among strings,
 that is, they focus on how to efficiently represent or encode common (repetitive) parts in strings.
However, none of them support linear-time pattern search.
Moreover, their pattern search time do not depend on only the pattern length
 but also the text length,
 and some indexes require (somewhat complicated) auxiliary data structures
 to improve pattern search time.
In short, those data structures achieve smaller indexes
by sacrificing pattern search time.

\partitle{Our contribution}
In this paper, we propose an efficient index for similar strings
without sacrifice the pattern search time.
It is  a novel data structure for similar strings,
 named {\em suffix tree of alignment}.
%
We assume that strings (texts) are aligned with each others, e.g.,
 two strings $A$ and $B$ can be represented as
 $\alpha_1 \beta_1 \ldots \alpha_k \beta_k \alpha_{k+1}$ and
 $\alpha_1 \delta_1 \ldots \alpha_k \delta_k \alpha_{k+1}$, respectively,
 where $\alpha_i$'s are common chunks and
  $\beta_i$'s and $\delta_i$'s are chunks different from the other string.
(We note that the given alignment is not required to be optimal.)
Then, our suffix tree for $A$ and $B$
 has the following properties.
(It should be noted that our index and algorithms can be generalized to three or more strings,
 although we only describe our contribution for two strings
for simplicity.)
\bi
\item {\bf Space reduction}:
Our suffix tree has $|A| + l_d + l_1$ leaves
 where $l_d$ is the sum of the lengths of {\em all} chunks of $B$ different from $A$
 (i.e., $\Sigma_{i=0}^{k} |\delta_i|$)
 and $l_1$ is the sum of the lengths of {\em some} common chunks of $A$ and $B$.
More precisely, $l_1$ is $\Sigma_{i=0}^{k} |\alpha^{*}_i|$
 where $\alpha^{*}_i$ is the longest suffix of $\alpha_i$ appearing at least twice in $A$ or in $B$.
The value of $\alpha^{*}_i$ is $O(\log \max(|A|,|B|))$ on average for random
 strings~\cite{karlin1983new}.
Furthermore, the values of $l_d$ and $l_1$ are very small in practice.
For instance, consider
two human genome sequences from two different individuals.
Since they are more than 99\% identical,
$l_d$ is very small compared to $|B|$.
We have computed $\alpha^{*}_i$ for human genome sequences
and found out $\alpha^{*}_i$ is very close to $\log \max(|A|,|B|)$,
 even though human genome sequences are not random.
Hence, our suffix tree is space-efficient for similar strings.
Note that the space of our index can be further reduced
 in the form of compressed indexes
 such as the compressed suffix tree~\cite{Grossi&Vitter:05,Sadakane07TCS}.
Our index is an important building block (rather than a final product)
 towards the goal of efficient indexing for highly similar data.
\item {\bf Pattern search}:
Our index is achieved without compromising the linear-time pattern search.
That is, using our suffix tree, one can search a pattern $P$ in $O(|P|+occ)$ time,
 where $occ$ is the number of occurrences of $P$ in $A$ and $B$.
In addition to the linear-time pattern search,
 we believe that our index supports the most of suffix tree functionalities,
 e.g., regular expression matchings, matching statistics, approximate matchings,
 substring range reporting,
 and so on~\cite{Baeza-Yates&Gonnet:96,Bille&Gortz:11,Gusfield:97},
 because our index is a kind of suffix trees.
\ei

\partitle{contribution 2}

We also present an efficient algorithm to construct the suffix tree of alignment.
One na\"{i}ve method to construct our suffix tree is constructing the generalized suffix tree
 and deleting unnecessary leaves.
However, it is not time/space-efficient.
\bi
\item When our suffix tree for the strings $A$ and $B$ is constructed from scratch,
  our construction algorithm requires $O(|A| + l_d + l_1 + l_2)$ time
  where $l_2$ is the sum of the lengths of other parts of common chunks of $A$ and $B$.
 More precisely, $l_2$ is $\Sigma_{i=1}^{k+1} |\hat{\alpha}_i|$
  where $\hat{\alpha}_i$ is the longest prefix of $\alpha_i$
  such that $d_i \alpha_i$ appears at least twice in $A$ and $B$
  ($d_i$ is the character preceding $\alpha_i$ in $B$.
Likewise with $l_1$,
 the value of $l_2$ is also very small compared to $|A|$ or $|B|$ in practice.
\item Our algorithm is incremental,
    i.e., we construct the suffix tree of $A$
     and then transform it to the suffix tree of the alignment.
    Thus, when the suffix tree of $A$ is already given,
     it requires $O(l_d + l_1 + l_2)$ time.
    $O(l_d + l_1 + l_2)$ is the minimum time required to make our index a kind of suffix tree
     so that linear-time pattern search is possible on both $A$ and $B$.
    Furthermore, our algorithm can be applied to the case when some strings are newly inserted
     or deleted.
\item Our algorithm uses constant-size extra working space except for our suffix tree itself.
 Thus, it is space-efficient compared to the na\"{i}ve method.
\ei
The space/time-efficiency of our construction algorithm becomes large
 when handling many strings.
The efficiency is feasible when the alignment has been computed in advance,
 which is the case in some applications.
For instance, in the Next-Generation Sequencing,
 the reference genome sequence is given
 and the genome sequence of a new individual is obtained
 by aligning against the reference sequence.
So, when a string (a new genome sequence) is obtained,
 the alignment is readily available.
Moreover, since our index does not require that the given alignment is optimal,
 we can use a near-optimal alignment instead of the optimal alignment
 if the time to compute an alignment is an important issue.
Since the given strings are assumed to be highly similar,
 a near-optimal alignment can be computed fast from exact string matching
 instead of dynamic programming requiring much time.

%
%
%


\section{Preliminaries}
\label{sec:preliminaries}



\subsection{Suffix trees}
\label{subsec:ST}


\partitle{Strings}

Let $S$ be a string over a fixed alphabet $\Sigma$.
A substring of $S$ beginning at the first position of $S$
 is called a {\it prefix} of $S$
 and a substring ending at the last position of $S$
 is called a {\it suffix} of $S$.
We denote by $|S|$ the length of $S$.
We assume that the last character of $S$ is a special symbol $\# \in \Sigma$,
 which occurs nowhere else in $S$.
%

%
%
\begin{figure}[t]
\centerline{\def\h{-2em}
\begin{tikzpicture}[font=\small,
                    inode/.style={text=black,inner sep=2pt,circle,draw=black},
                    leaf/.style={draw=black,text=black,inner sep=2pt, rectangle}]
    \node[inode] (root) at (3em,0em) {};
    \node[inode] (a) at (-4em,\h) {};
    \node[inode] (b) at (8em,\h) {};

    \draw[dashed] (root) to node[above,midway,sloped] {\texttt{b}} (b);
    \draw[dashed] (b) --+ (-2em,-2em) ;
    \draw[dashed] (b) --+ (+2em,-2em) ;

    \node[inode] (aa) at (-8em,2*\h) {};
    \node[leaf] (ad) at (-4em,2*\h) {};
    \node[inode] (ab) at (2em,2*\h) {};

    \node[inode] (aaab) at (-11em,3*\h) {};
    \node[inode] (aab) at (-5em,3*\h) {};
    \node[inode] (aba) at (0em,3*\h) {};
    \node[leaf] (abbaabad) at (3em,6*\h) {};


    \node[inode] (aaba) at (-7em,3.5*\h) {};
    \node[leaf] (aaabaaabbaaba) at (-12em,6*\h) {};
    \node[leaf] (aaabbaaba) at (-10em,6*\h) {};
    \node[leaf] (aabbaaba) at (-8em,6*\h) {};
    \node[leaf] (abad) at (1em,6*\h) {};
    \node[leaf] (aabad) at (-6em,6*\h) {};
    \node[leaf] (abaaabbaabad) at (-1em,6*\h) {};
    \node[leaf] (aabbaabad) at (-4em,6*\h) {};


    \draw (root) to node[above,midway,sloped] {\texttt{a}} (a) ;

    \draw (a) to node[above,midway,sloped] {\texttt{a}} (aa) ;
    \draw (a) to node[above,midway,right] {\texttt{\#}} (ad) ;
    \draw (a) to node[above,midway,sloped] {\texttt{b}} (ab);


    \draw (ab) to node[above,midway,sloped] {\texttt{baaba\#}} (abbaabad) ;
    \draw (ab) to node[above,midway,sloped] {\texttt{a}} (aba) ;

    \draw (aa) to node[above,midway,sloped] {\texttt{ab}} (aaab) ;
    \draw (aa) to node[above,midway,sloped] {\texttt{b}} (aab) ;

    \draw (aab) to node[above,midway,sloped] {\texttt{a}} (aaba) ;

    \draw (aba) to node[above,midway,right] {\texttt{\#}} (abad) ;

    \draw (aaab) to node[above,midway,sloped] {\texttt{aaabbaaba\#}} (aaabaaabbaaba);
    \draw (aaab) to node[above,midway,sloped] {\texttt{baaba\#}} (aaabbaaba);
    \draw (aaba) to node[above,midway,sloped] {\texttt{aabbaaba\#}} (aabbaaba);

    \draw (aba) to node[above,midway,sloped] {\texttt{aabbaaba\#}} (abaaabbaabad);

     \draw (aab) to node[above,midway,sloped] {\texttt{baaba\#}}(aabbaabad);

     \draw (aaba) to node[above,midway,right] {\texttt{\#}}(aabad);

\end{tikzpicture}}
\caption{The suffix tree of string {\tt aaabaaabbaaba\#}.
\label{fig:ST}}
\end{figure}
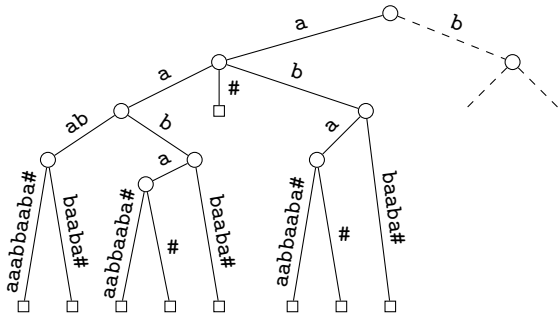

\partitle{Suffix trees}

The {\em suffix tree} of a string $S$ is a compacted trie with $|S|$ leaves,
 each of which represents each suffix of $S$.
Figure~\ref{fig:ST} shows the suffix tree of a string {\tt aaabaaabbaaba\#}.
For formal descriptions,
 the readers are referred to~\cite{Crochemore&Rytter:02,Gusfield:97}.
McCreight~\cite{McCreight:76} proposed a linear-time construction algorithm
 using auxiliary links called {\em suffix links}
 and also an algorithm for an {\em incremental editing},
 which transform the suffix tree of $S=\alpha \beta \gamma$
 to that of $S'=\alpha \delta \gamma$
 for some (possibly empty) string $\alpha$, $\beta$, $\delta$, $\gamma$.

%
%
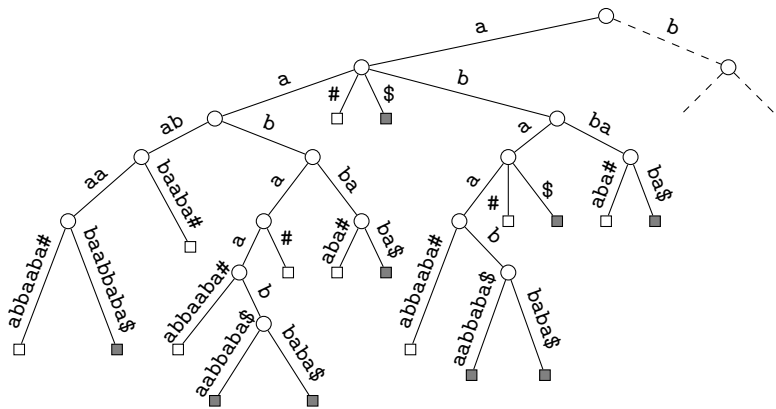
\begin{figure}[t]
\centerline{\def\h{-2.1em}
\begin{tikzpicture}[font=\small,
                    inode/.style={text=black,inner sep=2pt,circle,draw=black},
                    gleaf/.style={draw=black,fill=gray,text=black,inner sep=2pt, rectangle},
                    leaf/.style={draw=black,text=black,inner sep=2pt, rectangle}]

    \node[inode] (root) at (10em,0em) {};

    \node[inode] (a) at (0em,\h) {};
    \draw (root) to node[above,midway,sloped] {\texttt{a}} (a);

    \node[inode] (aa) at (-6em,2*\h) {};
    \node[leaf] (ad1) at (-1em,2*\h) {};
    \node[gleaf] (ad2) at (1em,2*\h) {};
    \node[inode] (ab) at (8em,2*\h) {};
    \draw (a) to node[above,midway,sloped] {\texttt{a}} (aa);
    \draw (a) to node[left,midway] {\texttt{\#}} (ad1);
    \draw (a) to node[midway,right] {\texttt{\$}} (ad2);
    \draw (a) to node[above,midway,sloped] {\texttt{b}} (ab);

    \node[inode] (aaab) at (-9em,2.75*\h) {};
    \node[inode] (aab) at (-2em,2.75*\h) {};
    \draw (aa) to node[above,midway,sloped] {\texttt{ab}} (aaab);
    \draw (aa) to node[above,sloped,midway] {\texttt{b}} (aab);

    \node[inode] (aba) at (6em,2.75*\h) {};
    \node[inode] (abba) at (11em,2.75*\h) {};
    \draw (ab) to node[above,midway,sloped] {\texttt{a}} (aba);
    \draw (ab) to node[above,sloped,midway] {\texttt{ba}} (abba);

    \node[inode] (aaabaa) at (-12em,4*\h) {};
    \node[leaf] (aaabbaabad1) at (-7em,4.5*\h) {};
    \draw (aaab) to node[above,midway,sloped] {\texttt{aa}} (aaabaa);
    \draw (aaab) to node[above,midway,sloped] {\texttt{baaba\#}} (aaabbaabad1);

    \node[inode] (aaba) at (-4em,4*\h) {};
    \node[inode] (aabba) at (0em,4*\h) {};
    \draw (aab) to node[above,midway,sloped] {\texttt{a}} (aaba);
    \draw (aab) to node[above,midway,sloped] {\texttt{ba}} (aabba);

    \node[inode] (abaa) at (4em,4*\h) {};
    \node[leaf] (abad1) at (6em,4*\h) {};
    \node[gleaf] (abad2) at (8em,4*\h) {};
    \draw (aba) to node[above,midway,sloped] {\texttt{a}} (abaa);
    \draw (aba) to node[left,near end] {\texttt{\#}} (abad1);
    \draw (aba) to node[right] {\texttt{\$}} (abad2);

    \node[leaf] (abbaabad1) at (10em,4*\h) {};
    \node[gleaf] (abbabad2) at (12em,4*\h) {};
    \draw (abba) to node[above,midway,sloped] {\texttt{aba\#}} (abbaabad1);
    \draw (abba) to node[above,midway,sloped] {\texttt{ba\$}} (abbabad2);

    \node[leaf] (aaabaaabbaabad1) at (-14em,6.5*\h) {};
    \node[gleaf] (aaabaabaabbabad2) at (-10em,6.5*\h) {};
    \draw (aaabaa) to node[above,midway,sloped] {\texttt{abbaaba\#}} (aaabaaabbaabad1);
    \draw (aaabaa) to node[above,midway,sloped] {\texttt{baabbaba\$}} (aaabaabaabbabad2);

    \node[inode] (aabaa) at (-5em,5*\h) {};
    \node[leaf] (aabad1) at (-3em,5*\h) {};
    \draw (aaba) to node[above,midway,sloped] {\texttt{a}} (aabaa);
    \draw (aaba) to node[right,near start] {\texttt{\#}} (aabad1);

    \node[leaf] (aabbabad1) at (-1em,5*\h) {};
    \node[gleaf] (aabbabad2) at (1em,5*\h) {};
    \draw (aabba) to node[above,midway,sloped] {\texttt{aba\#}} (aabbabad1);
    \draw (aabba) to node[above,midway,sloped] {\texttt{ba\$}} (aabbabad2);

    \node[leaf] (abaaabbaabad1) at (2em,6.5*\h) {};
    \node[inode] (abaab) at (6em,5*\h) {};
    \draw (abaa) to node[above,midway,sloped] {\texttt{abbaaba\#}} (abaaabbaabad1);
    \draw (abaa) to node[above, midway,sloped] {\texttt{b}} (abaab);

    \node[gleaf] (abaabaabbabad2) at (4.5em,7*\h) {};
    \node[gleaf] (abaabbabad2) at (7.5em,7*\h) {};
    \draw (abaab) to node[above,midway,sloped] {\texttt{aabbaba\$}} (abaabaabbabad2);
    \draw (abaab) to node[above,midway,sloped] {\texttt{baba\$}} (abaabbabad2);

    \node[leaf] (aabaaabbaabad1) at (-7.5em,6.5*\h) {};
    \node[inode] (aabaab) at (-4em,6*\h) {};
    \draw (aabaa) to node[above,midway,sloped] {\texttt{abbaaba\#}} (aabaaabbaabad1);
    \draw (aabaa) to node[above,midway,sloped] {\texttt{b}} (aabaab);

    \node[gleaf] (aabaabaabbabad2) at (-6em,7.5*\h) {};
    \node[gleaf] (aabaabbabad2) at (-2em,7.5*\h) {};
    \draw (aabaab) to node[above,midway,sloped] {\texttt{aabbaba\$}} (aabaabaabbabad2);
    \draw (aabaab) to node[above,midway,sloped] {\texttt{baba\$}} (aabaabbabad2);

    \node[inode] (b) at (15em,\h) {};
    \draw[dashed] (root) to node[above,midway,sloped] {\texttt{b}} (b);

    \draw[dashed] (b) --+ (-2em,-2em) ;
    \draw[dashed] (b) --+ (+2em,-2em) ;

\end{tikzpicture}}
\caption{The generalized suffix tree of two strings
  $A=$ {\tt aaabaaabbaaba\#} and $B=$ {\tt aaabaabaabbaba\$}.
 Leaves denoted by white squares and gray squares
  represent suffixes of $A$ and $B$, respectively.
\label{fig:GST}}
\end{figure}

\partitle{Generalized ST}

The {\em generalized suffix tree} of two strings $A$ and $B$
 is a suffix tree representing all suffixes of the two strings.
It can be obtained
 by constructing the suffix tree of the concatenated string $AB$
 where it is assumed that the last characters of $A$ and $B$
  are distinct~\cite{Amir&Farach:94,Gusfield:97}.
Thus, the generalized suffix tree has $|A|+|B|$ leaves
 and can be constructed in $O(|A|+|B|)$ time.
Figure~\ref{fig:GST} shows the generalized suffix tree of
 two strings {\tt aaabaaabbaaba\#} and {\tt aaabaabaabbaba\$}.


\subsection{Alignments}
\label{subsec:Alignments}


\partitle{definition}

Given two strings $A$ and $B$,
 an alignment of $A$ and $B$ is a mapping between the two strings that
 represents how $A$ can be transform to $B$ by replacing substrings of $A$ into those of $B$.
For example,
 let $A= \alpha \beta \gamma$ and $B= \alpha \delta \gamma$
 for some strings $\alpha$, $\beta$, $\gamma$, and $\delta$ ($\neq \beta$).
Then, we can get $B$ from $A$ by replacing $\beta$ with $\delta$.
We denote this replacement by alignment $\alpha (\beta /\, \delta) \gamma$.

\partitle{general form}

More generally, an alignment of two strings
 $A=\alpha_1 \beta_1 \ldots \alpha_k \beta_k \alpha_{k+1}$ and
 $B= \alpha_1 \delta_1 \ldots \alpha_k \delta_k \alpha_{k+1}$, for some $k\ge 1$,
 can be denoted by
 $\alpha_1 (\beta_1 /\, \delta_1) \ldots \alpha_k (\beta_k /\, \delta_k) \alpha_{k+1}$.
For simplicity, we assume that
 both $A$ and $B$ end with the special symbol $\# \in \Sigma$,
 which is contained in $\alpha_{k+1}$.
Without loss of generality,
 we assume the following conditions are satisfied
 for every $i=1,\ldots, k$.
\bi
\item $\alpha_{i+1}$ is not empty ($\alpha_1$ can be empty).
\item Either $\beta_i$ or $\delta_i$ can be empty.
\item The first characters of $\beta_i\alpha_{i+1}$ and $\delta_i\alpha_{i+1}$
       are distinct.
\ei
Note that these conditions are satisfied for the optimal alignments
 by most of popular distance measures such as the edit distance~\cite{Levenshtein:66}.
Moreover, alignments unsatisfying the conditions
 can be easily converted to satisfy the conditions.
If $\alpha_{i+1}$ ($i=1,\ldots,k-1$) is empty,
 $\beta_{i}$ and $\beta_{i+1}$ ($\delta_{i}$ and $\delta_{i+1}$) can be merged.
 (Note that $\alpha_{k+1}$ cannot be empty since $\#$ is contained in $\alpha_{k+1}$.)
If both $\beta_i$ and $\delta_i$ are empty,
 $\alpha_{i}$ and $\alpha_{i+1}$ can be merged.
Finally, if the first characters of $\beta_i\alpha_{i+1}$ and $\delta_i\alpha_{i+1}$
       are identical (say $c$),
 we include $c$ in $\alpha_i$ instead of $\beta_i\alpha_{i+1}$ and $\delta_i\alpha_{i+1}$.




\section{Suffix tree of simple alignments}
\label{sec:ST-simple-A}


\partitle{Intro.}

In this section,
 we define the suffix tree of a simple alignment ($k=1$)
 and present how to construct the suffix tree.


\subsection{Definitions}
\label{subsec:def-simple}


\partitle{a-suffix}

For some strings $\alpha$, $\beta$, $\gamma$, and $\delta$,
 let $\alpha (\beta /\, \delta)\gamma$ be an alignment of
 two strings $A= \alpha \beta \gamma$ and $B= \alpha \delta \gamma$.
We define suffixes of the alignment, called {\em alignment-suffixes} (for short {\em a-suffixes}).
Let $\alpha^{a}$ and $\alpha^{b}$ be the longest suffixes of $\alpha$
 which occur at least twice in $A$ and $B$, respectively,
 and let $\alpha^{*}$ be the longer of $\alpha^{a}$ and $\alpha^{b}$.
That is, $\alpha^{*}$ is the longest suffix of $\alpha$
 which occurs at least twice in $A$ or in $B$.
Then, there are 4 types of a-suffixes as follows.
\be
\item a suffix of $\gamma$,
\item a suffix of $\alpha^{*} \beta \gamma$ longer than $\gamma$,
\item a suffix of $\alpha^{*} \delta \gamma$ longer than $\gamma$,
\item $\alpha' (\beta /\, \delta) \gamma$
        where $\alpha'$ is a suffix of $\alpha$ longer than $\alpha^{*}$.
       (Note that an a-suffix of this type represents two normal suffixes
        derived from $A$ and $B$.)
\ee
For example, consider an alignment {\tt aaabaa(abba/baabb)aba\#}.
Then, $\alpha^{a}$ and $\alpha^{b}$ are {\tt baa} and {\tt aabaa}, respectively,
 and $\alpha^{*}$ is {\tt aabaa}.
Since $\alpha^{*}$ is {\tt aabaa},
 {\tt ba\#}, {\tt abaa\underline{abba}aba\#}, {\tt \underline{aabb}aba\#},
 and {\tt aaabaa(\underline{abba/baabb})aba\#} are a-suffixes of type 1, 2, 3, and 4, respectively
 (underlined strings denote symbols in $\beta$ and $\delta$).
The reason why we divide a-suffixes longer than $(\beta /\, \delta)\gamma$
 into ones longer than $\alpha^{*}(\beta /\, \delta)\gamma$ (type 4)
 and the others (types 2 and 3),
 or why $\alpha^*$ becomes the division point,
 has to do with properties of suffix trees
 and we explain the reason later.

%
%
\begin{figure}[t]
\centerline{\def\h{-2.1em}
\begin{tikzpicture}[font=\small,
                    inode/.style={text=black,inner sep=2pt,circle,draw=black},
                    bleaf/.style={draw=black,fill=black,text=black,inner sep=2pt, rectangle},
                    gleaf/.style={draw=black,fill=gray,text=black,inner sep=2pt, rectangle},
                    leaf/.style={draw=black,text=black,inner sep=2pt, rectangle}]

    \node[inode] (root) at (10em,0em) {};

    \node[inode] (a) at (0em,\h) {};
    \draw (root) to node[above,midway,sloped] {\texttt{a}} (a);

    \node[inode] (aa) at (-6em,2*\h) {};
    \node[bleaf] (ad2) at (0em,2*\h) {};
    \node[inode] (ab) at (8em,2*\h) {};
    \draw (a) to node[above,midway,sloped] {\texttt{a}} (aa);
    \draw (a) to node[midway,right] {\texttt{\#}} (ad2);
    \draw (a) to node[above,midway,sloped] {\texttt{b}} (ab);

    \node[inode] (aaab) at (-9em,3*\h) {};
    \node[inode] (aab) at (-2em,3*\h) {};
    \draw (aa) to node[above,midway,sloped] {\texttt{ab}} (aaab);
    \draw (aa) to node[above,sloped,midway] {\texttt{b}} (aab);

    \node[inode] (aba) at (6em,3*\h) {};
    \node[inode] (abba) at (11em,3*\h) {};
    \draw (ab) to node[above,midway,sloped] {\texttt{a}} (aba);
    \draw (ab) to node[above,sloped,midway] {\texttt{ba}} (abba);

    \node[bleaf] (aaabaa) at (-12em,7*\h) {};
    \node[leaf] (aaabbaabad1) at (-8em,5*\h) {};
    \draw (aaab) to node[above,midway,sloped] {\texttt{aa(abba/baabb)aba\#}} (aaabaa);
    \draw (aaab) to node[above,midway,sloped] {\texttt{baaba\#}} (aaabbaabad1);

    \node[inode] (aaba) at (-4em,4*\h) {};
    \node[inode] (aabba) at (0em,4*\h) {};
    \draw (aab) to node[above,midway,sloped] {\texttt{a}} (aaba);
    \draw (aab) to node[above,midway,sloped] {\texttt{ba}} (aabba);

    \node[inode] (abaa) at (4em,4*\h) {};
    \node[bleaf] (abad2) at (8em,4*\h) {};
    \draw (aba) to node[above,midway,sloped] {\texttt{a}} (abaa);
    \draw (aba) to node[right] {\texttt{\#}} (abad2);

    \node[leaf] (abbaabad1) at (10em,4*\h) {};
    \node[gleaf] (abbabad2) at (12em,4*\h) {};
    \draw (abba) to node[above,midway,sloped] {\texttt{aba\#}} (abbaabad1);
    \draw (abba) to node[above,midway,sloped] {\texttt{ba\#}} (abbabad2);


    \node[inode] (aabaa) at (-5em,5*\h) {};
    \node[leaf] (aabad1) at (-3em,5*\h) {};
    \draw (aaba) to node[above,midway,sloped] {\texttt{a}} (aabaa);
    \draw (aaba) to node[right,near start] {\texttt{\#}} (aabad1);

    \node[leaf] (aabbabad1) at (-1em,5*\h) {};
    \node[gleaf] (aabbabad2) at (1em,5*\h) {};
    \draw (aabba) to node[above,midway,sloped] {\texttt{aba\#}} (aabbabad1);
    \draw (aabba) to node[above,midway,sloped] {\texttt{ba\#}} (aabbabad2);

    \node[leaf] (abaaabbaabad1) at (2em,7*\h) {};
    \node[inode] (abaab) at (6em,5*\h) {};
    \draw (abaa) to node[above,midway,sloped] {\texttt{abbaaba\#}} (abaaabbaabad1);
    \draw (abaa) to node[above, midway,sloped] {\texttt{b}} (abaab);

    \node[gleaf] (abaabaabbabad2) at (4.5em,7*\h) {};
    \node[gleaf] (abaabbabad2) at (7.5em,7*\h) {};
    \draw (abaab) to node[above,midway,sloped] {\texttt{aabbaba\#}} (abaabaabbabad2);
    \draw (abaab) to node[above,midway,sloped] {\texttt{baba\#}} (abaabbabad2);

    \node[leaf] (aabaaabbaabad1) at (-8em,6.5*\h) {};
    \node[inode] (aabaab) at (-4em,6*\h) {};
    \draw (aabaa) to node[above,midway,sloped] {\texttt{abbaaba\#}} (aabaaabbaabad1);
    \draw (aabaa) to node[above,midway,sloped] {\texttt{b}} (aabaab);

    \node[gleaf] (aabaabaabbabad2) at (-6em,7.5*\h) {};
    \node[gleaf] (aabaabbabad2) at (-2em,7.5*\h) {};
    \draw (aabaab) to node[above,midway,sloped] {\texttt{aabbaba\#}} (aabaabaabbabad2);
    \draw (aabaab) to node[above,midway,sloped] {\texttt{baba\#}} (aabaabbabad2);

    \node[inode] (b) at (15em,\h) {};
    \draw[dashed] (root) to node[above,midway,sloped] {\texttt{b}} (b);

    \draw[dashed] (b) --+ (-2em,-2em) ;
    \draw[dashed] (b) --+ (+2em,-2em) ;

\end{tikzpicture}}
\caption{The suffix tree of an alignment {\tt aaabaa(abba/baabb)aba\#}.
 Leaves denoted by black squares, white squares, gray squares, and black diamonds
  represent a-suffixes of types 1, 2, 3, and 4, respectively.
\label{fig:ST-A}}
\end{figure}

\partitle{suffix tree of the alignment}

The {\em suffix tree of alignment} $\alpha (\beta /\,\delta) \gamma$
 is a compacted trie representing all a-suffixes of the alignment.
Formally, the suffix tree $T$ for the alignment
 is a rooted tree satisfying the following conditions.
\be
\item Each nonterminal arc is labeled with a nonempty substring of $A$ or $B$.
\item Each terminal arc is labeled with a nonempty suffix of $\beta \gamma$ or $\delta \gamma$,
       or with $\alpha'(\beta /\, \delta) \gamma$, where $\alpha'$ is a nonempty suffix of $\alpha$.
\item Each internal node $v$ has at least two children
       and the labels of arcs from $v$ to its children begin with distinct symbols.
\ee
Figure~\ref{fig:ST-A} shows the suffix tree of the alignment
{\tt aaabaa(abba/baabb)aba\#}.

\partitle{Difference}

The differences from classic suffix trees of strings
 (including generalized suffix trees) are as follows.
To reduce space, we represent common suffixes of $A$ and $B$ with one leaf.
For example, there exists one leaf representing {\tt aba\#} in Figure~\ref{fig:ST-A}
 because {\tt aba\#} is common suffixes of $A$ and $B$ (type 1).
However, suffixes of $A$ and $B$ longer than $\gamma$ derived from $(\beta / \delta)\gamma$
 are not common
 and thus we deal with these suffixes separately (types 2 and 3).
For suffixes longer than $(\beta / \delta)\gamma$,
 we have two cases.
First, consider an a-suffix $\alpha' (\beta /\, \delta) \gamma$ (type 4)
 such that $\alpha'$ is a suffix of $\alpha$ longer than $\alpha^{*}$,
 e.g., {\tt aaabaa(abba/baabb)aba\#}.
Due to the definition of $\alpha^{*}$,
 $\alpha'$ appears only once in each of $A$ and $B$ (at the same position)
 and we can represent $\alpha' (\beta /\, \delta) \gamma$ with one leaf
 by considering the terminal arc connected to the leaf
 is labeled with an alignment not a string,
 e.g., the leftmost (black diamond) leaf in Figure~\ref{fig:ST-A}.
However, it cannot be applicable to an a-suffix $\alpha'' (\beta /\, \delta) \gamma$
 such that $\alpha''$ is a suffix of $\alpha^{*}$,
 e.g., {\tt abaa(abba/baabb)aba\#}.
Since $\alpha''$ appears at least twice in $A$ or in $B$,
 $(\beta /\, \delta)$ may not be contained in the label of one arc.
Thus, we represent the a-suffix by two leaves,
 one of which represents $\alpha'' \beta \gamma$ (type 2)
 and the other $\alpha'' \delta \gamma$ (type 3),
 e.g., leaf $x$ representing {\tt abaa\underline{abba}aba\#}
 and leaf $y$ representing {\tt abaa\underline{baabb}aba\#}
 in Figure~\ref{fig:ST-A}.

\partitle{pattern search}

Pattern search can be solved using the suffix tree of alignment
 in the same way as using suffix trees of strings
 except for handling terminal arcs labeled with alignments.
When we meet a terminal arc labeled with
 an alignment $\alpha' (\beta /\, \delta) \gamma$ during search,
 we first compare $\alpha'$ with the pattern
 and then decide which of $\beta$ and $\delta$ we compare with the pattern
 by checking the first symbols of $\beta \gamma$ and $\delta \gamma$.
This comparison is in fact similar to branching at nodes.



\subsection{Construction}
\label{subsec:const-simple}


%

\partitle{Intro}

We describe how to construct the suffix tree $T$ for an alignment.
We assume the suffix tree $T^A$ of string $A$ is given.
($T^A$ can be constructed in $O(|A|)$ time~\cite{McCreight:76,Ukkonen:95}.)
To transform $T^A$ into the suffix tree $T$ of the alignment,
 we should insert the suffixes of $B$ into $T^A$.
We divide the suffixes of $B$ into three groups:
 suffixes of $\gamma$, suffixes of $\alpha^{*}\delta \gamma$ longer than $\gamma$,
 and suffixes of $\alpha\delta\gamma$ longer than $\alpha^{*}\delta \gamma$,
 which correspond to a-suffixes of types 1, 3, and 4, respectively.
First, we do not have to do anything for a-suffixes of type 1.
The suffixes of type 1 (suffixes of $\gamma$)
 already exist in $T^A$
because these are common suffixes in $A$ and $B$.

\partitle{outline}

Inserting the suffixes of $B$ longer than $\gamma$ consists of three steps.
We {\em explicitly} insert the suffixes shorter than or equal to $\alpha^{*}\delta \gamma$
 (type 3) in Steps A and B,
 and {\em implicitly} insert the suffixes longer than $\alpha^{*}\delta \gamma$ (type 4) in Step C
 as follows.
\be
\item[A.] Find $\alpha^{a}$
         and insert the suffixes of $\alpha^{a}\delta \gamma$ longer than $\gamma$.
\item[B.] Find $\alpha^{*}$
         and insert the suffixes of $\alpha^{*} \delta \gamma$
             longer than $\alpha^{a}\delta \gamma$.
\item[C.] Insert {\em implicitly} the suffixes of $\alpha\delta\gamma$
         longer than $\alpha^{*}\delta \gamma$.
\ee

%
%
\begin{figure}[t]
\centerline{\def\h{-2em}
\begin{tikzpicture}[font=\small,
                    inode/.style={text=black,inner sep=2pt,circle,draw=black},
                    leaf/.style={draw=black,text=black,inner sep=2pt, rectangle},
                    gleaf/.style={draw=black,fill=gray,text=black,inner sep=2pt, rectangle}]
    \node[inode] (root) at (3em,0em) {};
    \node[inode] (a) at (-4em,\h) {};
    \node[inode] (b) at (8em,\h) {};

    \draw[dashed] (root) to node[above,midway,sloped] {\texttt{b}} (b);
    \draw[dashed] (b) --+ (-2em,-2em) ;
    \draw[dashed] (b) --+ (+2em,-2em) ;

    \node[inode] (aa) at (-8em,2*\h) {};
    \node[leaf] (ad) at (-4em,2*\h) {};
    \node[inode] (ab) at (2em,2*\h) {};

    \node[inode] (aaab) at (-11em,3*\h) {};
    \node[inode] (aab) at (-5em,3*\h) {};
    \node[inode] (aba) at (1em,3*\h) {};
    \node[inode] (abba) at (4.5em,3*\h) {};
    \node[leaf] (abbaabad) at (3.5em,4*\h) {};
    \node[gleaf] (abbabad) at (5.5em,4*\h) {};

    \node[inode] (aaba) at (-6.5em,3.5*\h) {};
    \node[leaf] (aaabaaabbaaba) at (-12em,6*\h) {};
    \node[leaf] (aaabbaaba) at (-10em,6*\h) {};
    \node[leaf] (abad) at (2em,4*\h) {};
    \node[leaf] (aabad) at (-5.5em,4.25*\h) {};
    \node[inode] (aabaa) at (-7.5em,4.25*\h) {};
    \node[leaf] (aabaaabbaabad) at (-8em,6*\h) {};
    \node[gleaf] (aabaabbabad) at (-6.5em,6*\h) {};
    \node[leaf] (abaaabbaabad) at (-1em,6*\h) {};
    \node[gleaf] (abaabbaabad) at (1em,6*\h) {};
    \node[inode] (abaa) at (0em,4*\h) {};
    \node[inode] (aabba) at (-3.5em,4*\h) {};
    \node[gleaf] (aabbabad) at (-2.5em,6*\h) {};
    \node[leaf] (aabbaabad) at (-4.5em,6*\h) {};

    \draw (root) to node[above,midway,sloped] {\texttt{a}} (a) ;

    \draw (a) to node[above,midway,sloped] {\texttt{a}} (aa) ;
    \draw (a) to node[above,midway,right] {\texttt{\#}} (ad) ;
    \draw (a) to node[above,midway,sloped] {\texttt{b}} (ab);

    \draw (ab) to node[above,midway,sloped] {\texttt{ba}} (abba) ;
    \draw (abba) to node[above,midway,sloped] {\texttt{aba\#}} (abbaabad) ;
    \draw (abba) to node[above,midway,sloped] {\texttt{ba\#}} (abbabad) ;
    \draw (ab) to node[above,midway,sloped] {\texttt{a}} (aba) ;

    \draw (aa) to node[above,midway,sloped] {\texttt{ab}} (aaab) ;
    \draw (aa) to node[above,midway,sloped] {\texttt{b}} (aab) ;

    \draw (aab) to node[above,midway,sloped] {\texttt{a}} (aaba) ;

    \draw (aba) to node[above,midway,right] {\texttt{\#}} (abad) ;

    \draw (aaab) to node[above,midway,sloped] {\texttt{aaabbaaba\#}} (aaabaaabbaaba);
    \draw (aaab) to node[above,midway,sloped] {\texttt{baaba\#}} (aaabbaaba);
    \draw (aaba) to node[above,midway,right] {\texttt{\#}} (aabad);
    \draw (aaba) to node[above,midway,sloped] {\texttt{a}} (aabaa);
    \draw (aabaa) to node[above,midway,sloped] {\texttt{abbaaba\#}} (aabaaabbaabad);
    \draw (aabaa) to node[above,midway,sloped] {\texttt{bbaba\#}} (aabaabbabad);

    \draw (aba) to node[above,midway,sloped] {\texttt{a}} (abaa);
    \draw (abaa) to node[above,midway,sloped] {\texttt{bbaba\#}} (abaabbaabad);
    \draw (abaa) to node[above,midway,sloped] {\texttt{abbaaba\#}} (abaaabbaabad);

     \draw (aab) to node[above,midway,sloped] {\texttt{ba}}(aabba);
     \draw (aabba) to node[above,midway,sloped] {\texttt{ba\#}}(aabbabad);
     \draw (aabba) to node[above,midway,sloped] {\texttt{aba\#}}(aabbaabad);

\end{tikzpicture}}
\caption{The tree when Step A is applied to the suffix tree of $A$ in Figure~\ref{fig:ST}.
\label{fig:ST-A-S1}}
\end{figure}

%
In Step A, we first find $\alpha^{a}$ in $T^A$ using the doubling technique
 in incremental editing of~\cite{McCreight:76} as follows.
For a string $\chi$, we call a leaf a {\em $\chi$-leaf}
 if the suffix represented by the leaf contains $\chi$ as a prefix.
Then, if and only if a string $\chi$ occurs at least twice in $A$,
 there are at least two $\chi$-leaves in $T^A$.
To find $\alpha^{a}$,
 we check for some suffixes $\alpha'$ of $\alpha$
 whether or not there are at least two $\alpha'$-leaves in $T^A$.
Let $\alpha_{(j)}$ be the suffix of $\alpha$ of length $j$.
We first check whether or not there are at least two $\alpha_{(j)}$-leaves
 in increasing order of $j=1,2,4,8,\ldots, |\alpha|$.
Suppose $\alpha_{(h)}$ is the shortest suffix among these $\alpha_{(j)}$'s 
 such that there is only one $\alpha_{(j)}$-leaf.
(Note that $h/2 \le |\alpha^{a}| < h $.)
Then, $\alpha^{a}$ can be found
 by checking whether or not there are at least two $\alpha_{(j)}$-leaves
 in decreasing order of $j=h-1, h-2, \ldots$,
 which can be done efficiently using suffix links~\cite{McCreight:76}.

After finding $\alpha^{a}$,
 we insert the suffixes from the longest $\alpha^{a}\delta \gamma$
 to the shortest $d\gamma$ where $d$ is the last character of $\alpha\delta$:
Inserting the longest suffix is done by traversing down the suffix tree from the root
 and inserting the other suffixes can be done efficiently
 using suffix links~\cite{McCreight:76,Ukkonen:95}.
Figure~\ref{fig:ST-A-S1} shows the tree when Step A is applied to
the suffix tree of $A$ in Figure~\ref{fig:ST}.


\partitle{Step B}

Let $T'$ be the tree when Step A is finished.
In Step B, we first find $\alpha^{*}$ using $T'$
 and insert into $T'$ the suffixes longer than $\alpha^{a}\delta\gamma$
 from the longest $\alpha^{*}\delta \gamma$ to the shortest
 in the same way as we did in Step A.
Unlike Step A, however,
 we have the following difficulties for finding $\alpha^{*}$
 because $T'$ is an incomplete suffix tree:
 i) suffixes of $B$ longer than $\alpha^{a}\delta\gamma$ are not represented in $T'$,
 ii) both suffixes of $A$ and suffixes of $B$ are represented in one tree $T'$, and
 iii) some suffixes of $B$ (a-suffixes of type 1)
     share leaves with suffixes of $A$
     but some suffixes (a-suffixes of type 3) of $B$ do not.

But we show that $T'$ has sufficient information to find $\alpha^{*}$.
(Recall $\alpha^{*}$ is the longest suffix of $\alpha$
 occurring at least twice in $A$ or in $B$.)
Notice that our goal in Step $B$ is finding $\alpha^{*}$ but not $\alpha^{b}$.
If $|\alpha^{b}| \le |\alpha^{a}|$,
 for no suffix $\alpha'$ of $\alpha$ longer than $\alpha^{a}$,
 there are at least two $\alpha'$-leaves in $T'$,
 in which case $\alpha^{*}=\alpha^{a}$.
Thus, we do not need to consider suffixes of $\alpha$ shorter than or equal to $\alpha^{a}$.

%
\begin{lemma} \label{lem:1}
For a suffix $\alpha'$ of $\alpha$ longer than $\alpha^{a}$,
 if and only if $\alpha'$ occurs at least twice in $B$,
 there are at least two $\alpha'$-leaves in $T'$.
\end{lemma}
\begin{proof}
(If) We first show that there is an $\alpha'$-leaf in $T'$
      due to the occurrence $occ_1$ of $\alpha'$ as a suffix of $\alpha$.
     Since $\alpha$ is common in $A$ and $B$,
      $occ_1$ appears in both $A$ and $B$
      as prefixes of $\alpha'\beta\gamma$ and $\alpha'\delta\gamma$, respectively.
     Note that $\alpha'\delta\gamma$ is not represented in $T'$
     but $\alpha'\beta\gamma$ is.
     Hence, there is an $\alpha'$-leaf $f_1$ in $T'$ due to $occ_1$.

    Next, we show that there is another $\alpha'$-leaf in $T'$
     due to an occurrence $occ_2$ of $\alpha'$ other than $occ_1$ in $B$.
    Let $p_1$ and $p_2$ be the start positions of $occ_1$ and $occ_2$ in $B$, respectively,
    and let $p_a$ be the start position of the suffix $\alpha^{a}\delta\gamma$ in $B$.
    %
    We first prove by contradiction that $occ_2$ is contained in $\alpha^{a}\delta\gamma$.
    Suppose otherwise, that is, $p_2$ precedes $p_a$.
    We have two cases according to which of $p_1$ and $p_2$ precedes.
    First consider the case that $p_2$ precedes $p_1$,
    In this case, $occ_2$ is properly contained in $\alpha$,
     which means that $\alpha'$ appears at least twice in $\alpha$
     and also in $A$.
    This contradicts with the definition of $\alpha^{a}$
     since $\alpha'$ is longer than $\alpha^{a}$.
    Consider the case that $p_1$ precedes $p_2$.
    Let $\alpha''$ be the suffix of $\alpha$ starting at $p_2$.
    Then, $|\alpha'| > |\alpha''| > |\alpha^{a}|$
     and $\alpha''$ is a prefix of $occ_2$.
    Furthermore, $\alpha''$ is also a prefix of $\alpha'\delta\gamma$.
    It means that $\alpha''$ occurs twice in $\alpha$
     as a proper prefix of $\alpha'$ and a proper suffix of $\alpha'$.
    This contradicts with the definition of $\alpha^{a}$
     since $\alpha''$ is longer than $\alpha^{a}$.
    Therefore, $p_2$ does not precede $p_a$,
     which means $occ_2$ is contained in $\alpha^{a}\delta\gamma$.

    Now we show that there is an $\alpha'$-leaf $f_2$ in $T'$ due to $occ_2$
     and $f_2$ is distinct from $f_1$.
    Let $\eta$ be the suffix of $B$ starting at position $p_2$.
    Then, $\eta$ is a proper suffix of $\alpha^{a}\delta\gamma$
     since $p_2$ follows $p_a$.
    Because $T'$ represents all suffixes of $\alpha^{a}\delta\gamma$,
     there exists an $\alpha'$-leaf $f_2$ representing $\eta$ in $T'$.
    Moreover, suffixes of $A$ and $B$ share leaves in $T'$
     only if they are suffixes of $\gamma$.
    Since the suffix $\alpha'\beta\gamma$ represented by $f_1$ is longer than $\gamma$,
    $f_1$ and $f_2$ are distinct.

(Only if) We prove by contradiction the converse,
        i.e., if $\alpha'$ occurs only once in $B$,
        there is only one $\alpha'$-leaf in $T'$.
    Suppose there are two $\alpha'$-leaves in $T'$.
    Since $\alpha'$ occurs only once in $B$,
     no suffix of $B$ except for $\alpha'\delta\gamma$ contains $\alpha'$ as a prefix.
    Moreover, there is no leaf representing $\alpha'\delta\gamma$ in $T'$
     because $|\alpha'\delta\gamma| > |\alpha^{a}\delta\gamma|$ and
     no suffix of $B$ longer than $\alpha^{a}\delta\gamma$ is represented in $T'$.
    Thus, no $\alpha'$-leaf in $T'$ represents a suffix of $B$
     and the two $\alpha'$-leaves in $T'$ represent two suffixes of $A$.
    It means $\alpha'$ occurs twice in $A$,
     which contradicts with the definition of $\alpha^{a}$
     that $\alpha^{a}$ is the longest suffix of $\alpha$
     occurring at least twice in $A$
     since $|\alpha'| > |\alpha^{a}|$.
    Therefore, there is only one $\alpha'$-leaf in $T'$
     if $\alpha'$ occurs only once in $B$.
\qed
\end{proof}

\begin{corollary}  \label{cor:1}
For a suffix $\alpha'$ of $\alpha$ longer than $\alpha^{a}$,
 if and only if $\alpha'$ occurs at least twice in $A$ or in $B$,
 there are at least two $\alpha'$-leaves in $T'$.
\end{corollary}
By Corollary~\ref{cor:1}, we can find $\alpha^{*}$
 by checking for some suffixes $\alpha'$ of $\alpha$ longer than $\alpha^{a}$
 whether or not there are at least two $\alpha'$-leaves in $T'$.
It can be done in $O(|\alpha^{*}|)$ using the way similar to Step A.
When Step B is applied to the tree in Figure~\ref{fig:ST-A-S1},
 the resulting tree is the same as the tree in Figure~\ref{fig:ST-A}
 except that the terminal arc connected to the leftmost leaf (black diamond)
  is labeled with suffix {\tt aaabbaaba\#} of string $A$
  but not with a-suffix {\tt aa(abba/baabb)aba\#} of the alignment.

\partitle{type 4}

In Step C, for every suffix $\alpha'$ of $\alpha$ longer than $\alpha^{*}$,
 we {\em implicitly} insert the suffix $\alpha' \delta \gamma$ of $B$.
Since the suffix $\alpha' \delta \gamma$ of $B$ and the suffix $\alpha'\beta\gamma$ of $A$
 (a-suffixes of type 4) should be represented by one leaf,
 we do not insert a new leaf but convert the leaf representing $\alpha' \beta \gamma$
 to represent the a-suffix $\alpha' (\beta /\, \delta)\gamma$.
It can be done by replacing {\em implicitly}
 every $\beta$ properly contained in labels of terminal arcs
 with $(\beta /\, \delta)$.
Consequently, we explicitly do nothing in Step C,
 and these implicit changes are already reflected in the given alignment.
For example, the suffix tree in Figure~\ref{fig:ST-A}
 is obtained by replacing implicitly the label {\tt aa\underline{abba}aba\#}
 of the terminal arc connected to the leftmost leaf (black diamond)
 with a-suffix {\tt aa(abba/baabb)aba\#} of the alignment.

\partitle{time complexity}

We consider the time complexity of our algorithm.
In step A, finding $\alpha^{a}$ takes $O(|\alpha^{a}|)$ time
 and inserting suffixes takes $O(|\alpha^{a}\delta\hat{\gamma}|)$ time,
 where $\hat{\gamma}$ is the longest prefix of $\gamma$ such that
 $d\hat{\gamma}$ occurs at least twice in $A$ and $B$
 (where $d$ is the character preceding $\gamma$).
For detailed analysis, the readers are referred to~\cite{McCreight:76}.
In step B, similarly, finding $\alpha^{*}$ takes $O(|\alpha^{*}|)$ time
 and inserting suffixes takes $O(|\alpha^{*}\delta\hat{\gamma}|)$ time.
Step C takes no time since it is implicitly done.
Thus, we get the following theorem.

\begin{theorem}
Given an alignment $\alpha (\beta /\, \delta) \gamma$
 and the suffix tree of string $\alpha \beta \gamma$,
 the suffix tree of $\alpha (\beta /\, \delta) \gamma$
 can be constructed in $O(|\alpha^{*}| + |\delta| + |\hat{\gamma}|)$ time.
\end{theorem}


\section{Suffix tree of general alignments}
\label{sec:ST-general-A}


\partitle{Input strings}
We extend the definitions and the construction algorithm
 into more general alignments.
Let $\alpha_1 (\beta_1 /\, \delta_1) \ldots \alpha_k (\beta_k /\, \delta_k) \alpha_{k+1}$
 be an alignment of two strings
 $A=\alpha_1 \beta_1 \ldots \alpha_k \beta_k \alpha_{k+1}$ and
 $B= \alpha_1 \delta_1 \ldots \alpha_k \delta_k \alpha_{k+1}$.
For $1\le i \le k+1$,
 let $\alpha^{a}_{i}$ and $\alpha^{b}_{i}$ be the longest suffixes of $\alpha_i$
 occurring at least twice in $A$ and $B$, respectively,
 and let $\alpha^{*}_{i}$ be the longer of $\alpha^{a}_{i}$ and $\alpha^{b}_{i}$.
That is, $\alpha^{*}_{i}$  is the longest suffix of $\alpha_{i}$
 which occurs at least twice in $A$ or in $B$.
Moreover, let $\hat{\alpha}_{i}$ be the longest prefix of $\alpha_{i}$ such that
 $d_{i}\hat{\alpha}_{i}$ occurs at least twice in $A$ and $B$
 where $d_{i}$ is the character preceding $\alpha_{i}$ in $B$.

The suffix tree of the alignment is a compacted trie
 that represents the following a-suffixes of the alignment.
\be
\item a suffix of $\alpha_{k+1}$,
\item a suffix of $\alpha^{*}_i \beta_i \alpha_{i+1}\ldots\alpha_{k+1}$
         longer than $\alpha_{i+1}\ldots\alpha_{k+1}$,
\item a suffix of $\alpha^{*}_i \delta_i \alpha_{i+1}\ldots\alpha_{k+1}$
         which is longer than $\alpha_{i+1}\ldots\alpha_{k+1}$,
\item $\alpha_{i}' (\beta_i /\, \delta_i) \ldots \alpha_{k+1}$,
        where $\alpha_{i}'$ is a suffix of $\alpha_{i}$ longer than $\alpha^{*}_{i}$.
\ee

Given the suffix tree of $A$,
 the suffix tree of the alignment can be constructed as follows
 (the details are omitted).
\be
\item[A1.] Find $\alpha^{a}_{i}$ using the suffix tree of $A$ for each $i$ ($1 \le i \le k$).
\item[A2.] Insert the suffixes of
         $\alpha^{a}_{i} \delta_i \alpha_{i+1} \ldots \alpha_{k+1}$
         longer than $\alpha_{i+1}\ldots\alpha_{k+1}$ for each $i$.
\item[B1.] Find $\alpha^{*}_{i}$ for each $i$.
\item[B2.] Insert the suffixes of
         $\alpha^{*}_{i} \delta_i  \ldots \alpha_{k+1}$
         longer than $\alpha^{a}_{i} \delta_i  \ldots \alpha_{k+1}$ for each $i$.
\item[C.] Insert {\em implicitly}
         the suffixes of $\alpha_{i} \delta_i  \ldots \alpha_{k+1}$
         longer than $\alpha^{*}_{i} \delta_i  \ldots \alpha_{k+1}$
         for each $i$.
\ee

\begin{theorem}
Given an alignment
 $\alpha_1 (\beta_1 /\, \delta_1) \ldots \alpha_k (\beta_k /\, \delta_k) \alpha_{k+1}$
 and the suffix tree of string
 $\alpha_1 \beta_1 \ldots \alpha_k \beta_k \alpha_{k+1}$,
 the suffix tree of the alignment
 can be constructed in time at most linear to
 the sum of the lengths of $\alpha^{*}_i$, $\delta_i$, $\hat{\alpha}_{i+1}$
 for $1 \le i \le k$.
\end{theorem}

\partitle{3 strings}

Our definitions and algorithms can be also extended into
 alignments of more than two strings.
For example, consider an alignment $\alpha (\beta /\, \delta /\,\vartheta) \gamma$
 of three strings $A=\alpha\beta\gamma$, $B=\alpha\delta\gamma$, and $C=\alpha\vartheta\gamma$
 such that the first characters of $\beta\gamma$, $\delta\gamma$, and $\vartheta\gamma$
 are distinct.
We define $\alpha^{a}$, $\alpha^{b}$, and $\alpha^{c}$ as
 the longest suffix of $\alpha$ which occurs at least twice in $A$, $B$, and $C$, respectively,
 and $\alpha^{*}$ as the longest of $\alpha^{a}$, $\alpha^{b}$, and $\alpha^{c}$.
Then, there are 5 types of a-suffixes.
(Suffixes of $\alpha^{*} \vartheta \gamma$ longer than $\gamma$
  are added as a new type of a-suffixes.)
The suffix tree of the alignment can be defined similarly and constructed as follows:
 From the suffix tree of $\alpha\beta\gamma$,
 we construct the suffix tree of $\alpha (\beta /\, \delta ) \gamma$
  by inserting some suffixes of $B$,
 and then convert into the suffix tree of $(\beta /\, \delta /\,\vartheta) \gamma$
  by inserting suffixes of $C$ (and some suffixes of $B$ occasionally).
We omit the details.



%
%




\end{document}